\documentclass{article}
\usepackage{amsmath,amssymb,amsthm}
\usepackage[noend]{algpseudocode}
\usepackage{algorithm}
\usepackage{color}
\usepackage{lineno}
\usepackage{graphicx}
\usepackage{ifthen}
\usepackage{hyperref}
\definecolor{gray}{rgb}{0.7,0.7,0.7}

\newcommand{\appendixMode}{0}
\newcommand{\appdx}[3]{\ifthenelse{\equal{#1}{0}}{#2}{#3}}
\newcommand{\arxivonly}[1]{#1}

\newcommand{\nerve}{\mathrm{Nrv}}
\newcommand{\pers}{\mathrm{Pers}}
\newcommand{\cech}{\mathcal{C}}
\newcommand{\rips}{\mathcal{R}}
\newcommand{\pP}{{\widetilde{P}}}
\newcommand{\e}{\varepsilon}

\newcommand{\R}{\mathbb{R}}

\newcommand{\nbr}{\mathrm{nbr}}

\newcommand{\parent}{\mathrm{parent}}
\newcommand{\ball}{\mathrm{ball}}

\newcommand{\dist}{\mathrm{\mathbf{d}}}

\newcommand{\link}{\mathrm{Lk\;}}

\renewcommand{\because}[1]{& \left[\text{#1}\right]}

\newcommand{\ch}{\mathrm{ch}}
\newcommand{\pred}{\mathrm{pred}}
\newcommand{\D}{\mathcal{D}}
\DeclareMathOperator*{\argmax}{arg\,max}

\newtheorem{theorem}{Theorem}
\newtheorem{lemma}[theorem]{Lemma}
\newtheorem{proposition}[theorem]{Proposition}
\newtheorem{corollary}[theorem]{Corollary}

\title{A Geometric Perspective on Sparse Filtrations\thanks{A short version of this paper appeared in the proceedings of the 2015 Canadian Conference on Computational Geometry}}

\author{
  Nicholas J. Cavanna\thanks{University of Connecticut \texttt{nicholas.j.cavanna@uconn.edu}}
  \and
  Mahmoodreza Jahanseir\thanks{University of Connecticut \texttt{reza@engr.uconn.edu}}
  \and
  Donald R. Sheehy\thanks{University of Connecticut \texttt{don.r.sheehy@gmail.com}}
}

\date{}

\begin{document}
  \maketitle

  \begin{abstract}

  We present a geometric perspective on sparse filtrations used in topological data analysis.
  This new perspective leads to much simpler proofs, while also being more general, applying equally to Rips filtrations and \v Cech filtrations for any convex metric.
  We also give an algorithm for finding the simplices in such a filtration and prove that the vertex removal can be implemented as a sequence of elementary edge collapses.
  
  A video illustrating this approach is available~\cite{cavanna15visualizing}\arxivonly{ as well as a short conference version~\cite{cavanna15geometric}}.
\end{abstract}
  \section{Introduction} 
\label{sec:introduction}

  Given a finite data set in a Euclidean space, it is natural to consider the balls around the data points as a way to fill in the space around the data and give an estimate of the missing data.
  The union of balls is often called the \emph{offsets} of the point set.
  Persistent homology was originally invented as a way to study the changes in topology of the offsets of a point set as the radius increases from $0$ to $\infty$.
  The input to persistent homology is usually a filtered simplicial complex, that is, an ordered collection of simplices (vertices, edges, triangles, etc.) such that each simplex appears only after its boundary simplices of one dimension lower.
  The Nerve Theorem and its persistent variant allow one to compute the persistent homology of the offsets by instead looking at a discrete object, a filtered simplicial complex called the nerve (see Fig.~\ref{fig:pipeline}).
  The simplest version of this complex is called the \v Cech complex and it may be viewed as the set of all subsets of the input, ordered by the radius of their smallest enclosing ball.
  Naturally, the \v Cech complex gets very big very fast, even when restricting to subsets of constant size.
  A common alternative is the Rips complex but it suffers similar difficulties.
  Over the last few years, there have been several approaches to building sparser complexes that still give good approximations to the persistent homology~\cite{sheehy13linear,kerber13approximate,dey14computing,buchet15efficient,botnan15approximating}.

   \begin{figure}
    \centering
      \appdx{\appendixMode}{
        \includegraphics[width=0.95\textwidth]{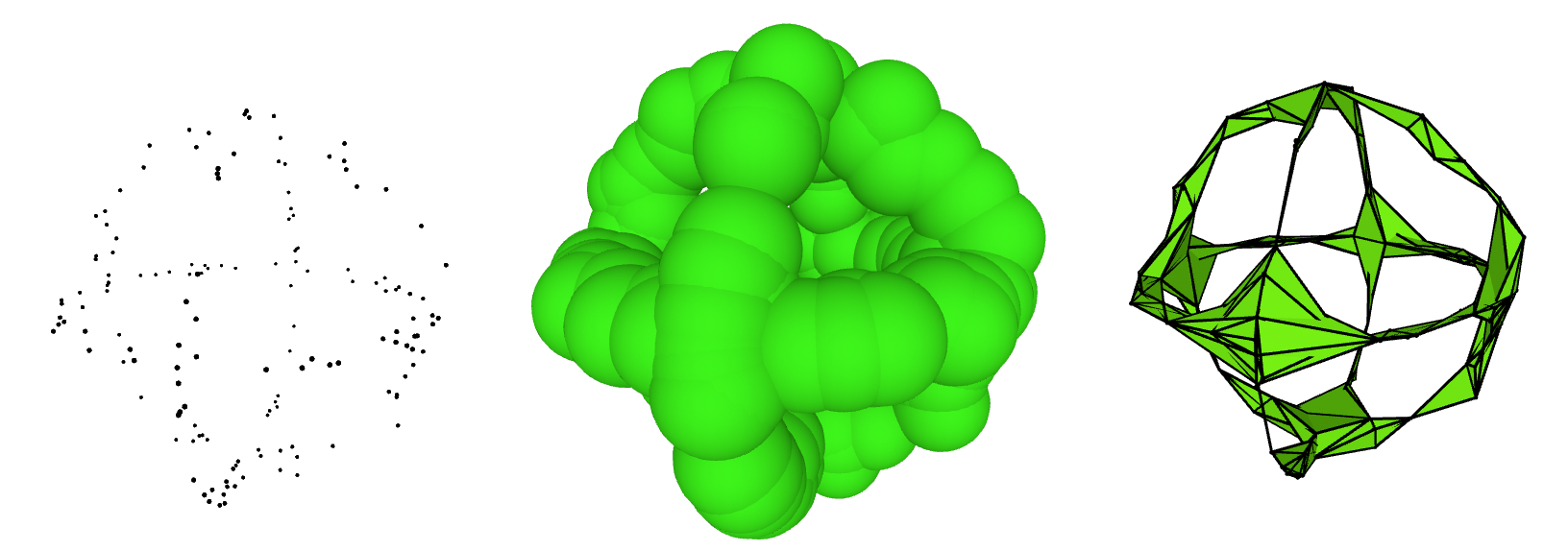}
      }{
        \includegraphics[width=0.45\textwidth]{figures/pipeline}
      }
        
    \caption{A point set sampled on a sphere, its offsets, and its (sparsified) nerve complex.}
    \label{fig:pipeline}
  \end{figure}

  Our main contributions are the following.
  \begin{enumerate}
    \item A much simpler explanation for the construction and proof of correctness of sparse filtrations.
    Our new geometric construction shows that the sparse complex is just a nerve in one dimension higher.
    \item The approach easily generalizes to Rips, \v Cech and related complexes (the offsets for any convex metric).
    This is another advantage of the geometric view as the main result follows from convexity rather than explicit construction of simplicial map homotopy equivalences.
    \item A simple geometric proof that the explicit removal of vertices from the sparse filtration can be done with simple edge contractions.
    This can be done without resorting to the full-fledged zig-zag persistence algorithm~\cite{carlsson09zigzag,carlsson10zigzag,maria15zigzag,milosavljevic11zigzag} or even the full simplicial map persistence algorithm~\cite{dey14computing,boissonat13compressed}.
  \end{enumerate}

  The most striking thing about this paper is perhaps more in what is absent than what is present.
  Despite giving a complete treatment of the construction, correctness, and approximation guarantees of sparse filtrations that applies to both \v Cech and Rips complexes, there is no elaborate construction of simplicial maps or proofs that they induce homotopy equivalences.
  In fact, we prove the results directly on the geometric objects, the covers, rather than the combinatorial objects, the complexes, and the result is much more direct.
  In a way, this reverses a common approach in computational geometry problems in which the geometry is as quickly as possible replaced with combinatorial structure; instead, we delay the transition from the offsets to a discrete representation until the very end of the analysis.


  \paragraph{Related Work.} 
  \label{par:related_work}

    Soon after the introduction of persistent homology by Edelsbrunner et al.~\cite{edelsbrunner02topological}, there was interest in building more elaborate complexes for larger and larger data sets.
    Following the full algebraic characterization of persistent homology by Zomorodian and Carlsson~\cite{zomorodian05computing}, a more general theory of zigzag persistence was developed~\cite{carlsson09zigzag,carlsson10zigzag,maria15zigzag,milosavljevic11zigzag} using a more complicated algorithm.
    Zig-zags gave a way to analyze spaces that did not grow monotonically; they could alternately grow and shrink such as by growing the scale and then removing points~\cite{tausz11applications}.
    A variant of this techniques was first applied for specific scales by Chazal and Oudot in work on manifold reconstruction~\cite{chazal08towards} and was implemented as a full zigzag by Morozov in his Dionysus library~\cite{dionysus}.
    Later, Sheehy gave a zig-zag for Rips filtrations that came with guaranteed approximation to the persistent homology of the unsparsified filtration~\cite{sheehy13linear}.
    Other later works gave various improvements and generalizations of sparse zig-zags~\cite{oudot14zigzag,kerber13approximate,dey14computing,botnan15approximating}.



  \section{Background} 
\label{sec:background}

  \paragraph{Distances and Metrics.} 
    Throughout, we will assume the input is a finite point set $P$ in $\R^d$ endowed with some convex metric $\dist$.
    A closed ball with center $c$ and radius $r$ will be written as $\ball(c,r) = \{x\in \R^d | \dist(x,c)\le r\}$.
    For illustrative purposes, we will often draw balls as Euclidean ($\ell_2$) balls.

    For a non-negative $\alpha \in \R$, the \emph{$\alpha$-offsets} of $P$ are defined as
    \[
      P^\alpha := \bigcup_{p\in P}\ball(p,\alpha).
    \]
    The sequence of offsets as $\alpha$ ranges from $0$ to $\infty$ is called the \emph{offsets filtration} $\{P^\alpha\}$.

    The \emph{doubling dimension} of a metric space is $\log_2 \gamma$, where $\gamma$ is the maximum over all balls $B$, of the minimum number of balls of half the radius of $B$ required to cover $B$.
    Metric spaces with a small constant doubling dimension are called \emph{doubling metrics}.
    Such metrics allow for packing arguments similar to those used in Euclidean geometry.
    \appdx{\appendixMode}{
    For example, consider the following simple exercise.
    If a set of points in a metric of doubling dimension $\rho$ are pairwise of distance at least $\epsilon$ apart and all contained in a ball of radius $c\epsilon$, then there are fewer than $(2c)^\rho$ points.
    }{}
    

  \paragraph{Simplicial Complexes.} 
    A \emph{simplicial complex} $K$ is a family of subsets of a vertex set that is closed under taking subsets.
    The sets $\sigma\in K$ are called \emph{simplices} and $|\sigma| - 1$ is called the \emph{dimension} of $\sigma$.
    A nested family of simplicial complexes is called a \emph{simplicial filtration}.
    Often the family of complexes will be parameterized by a nonnegative real number as in $\{K^\alpha\}_{\alpha\ge 0}$.
    Here, the filtration property guarantees that $\alpha\le \beta$ implies that $K^\alpha \subseteq K^\beta$.
    In this case, the value of $\alpha$ for which a simplex first appears is called its birth time, and so, if there is a largest complex $K^\alpha$ in the filtration, the whole filtration can be represented by $K^\alpha$ and the birth time of each simplex.
    For this reason, simplicial filtrations are often called \emph{filtered simplicial complex}.


  \paragraph{Persistent Homology.} 
    Homology is an algebraic tool for characterizing the connectivity of a space.
    It captures information about the connected components, holes, and voids.
    For this paper, we will only consider homology with field coefficients and the computations will all be on simplicial complexes.
    In this setting, computing homology is done by reducing a matrix $D$ called the boundary matrix of the simplicial complex.
    The boundary matrix has one row and column for each simplex.
    If the matrix reduction respects the order of a filtration, i.e.\ columns are only combined with columns to their left, then the reduced matrix also represents the so-called \emph{persistent homology} of the filtration.
    Persistent homology describes the changes in the homology as the filtration parameter changes and this information is often expressed in a \emph{barcode} (See Fig.~\ref{fig:barcode}).
    Barcodes give topological signatures of a shape~\cite{ghrist08barcodes}.

    \begin{figure}[htbp]
      \centering
      \appdx{\appendixMode}{
        \includegraphics[width=0.9\textwidth]{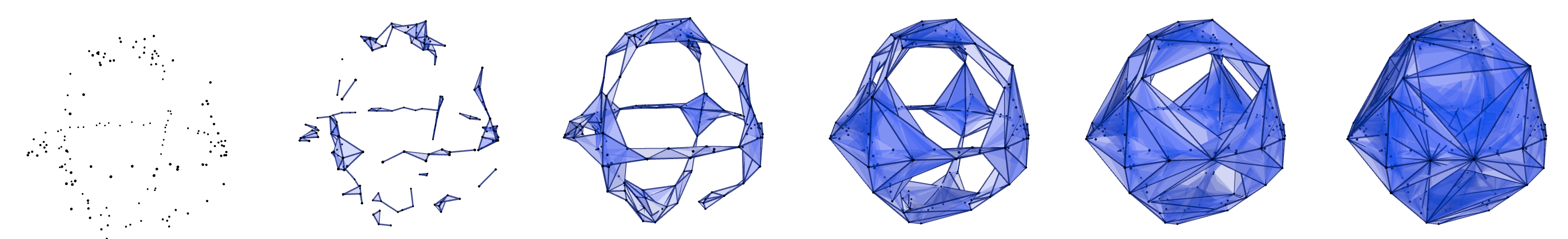}
        \includegraphics[width=0.9\textwidth]{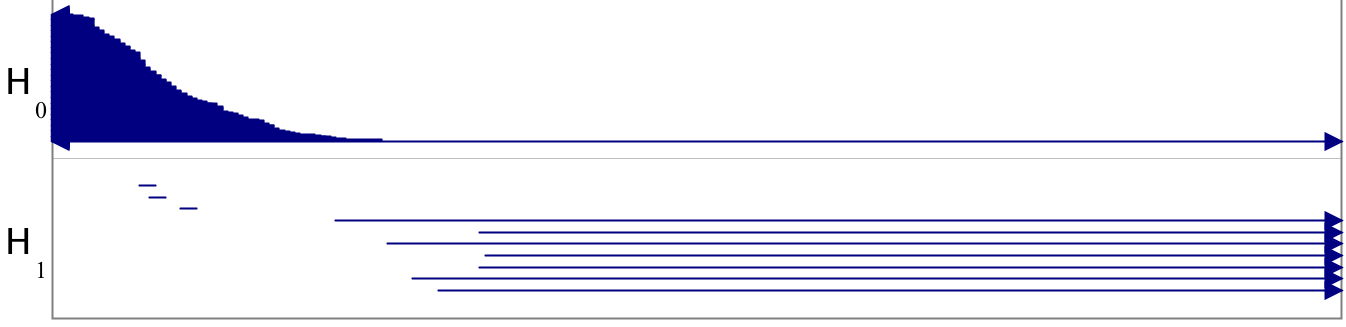}
      }{
        \includegraphics[width=0.45\textwidth]{figures/complexes}
        \includegraphics[width=0.45\textwidth]{figures/barcode}
      }
      \caption{A filtration and its barcode.}
      \label{fig:barcode}
    \end{figure}

    Each bar of a barcode is an interval encoding the lifespan of a topological feature in the filtration.
    We say that a barcode $B_1$ is a (multiplicative) $c$-approximation to another barcode $B_2$ if there is a partial matching between $B_1$ and $B_2$ such that every bar $[b,d]$ with $d/b>c$ is matched and every matched pair of bars $[b,d],[b',d']$ satisfies $\max\{b/b', b'/b, d/d', d'/d\}  \le c$.
    A standard result on the stability of barcodes~\cite{chazal09proximity} implies that if two filtrations $\{F^\alpha\}$ and $\{G^\alpha\}$ are $c$-interleaved in the sense that $F^{\alpha/c}\subseteq G^\alpha \subseteq F^{c\alpha}$, then the barcode of $\{F^\alpha\}$ is a $c$-approximation to $\{G^\alpha\}$.


  \paragraph{Nerve Complexes and Filtrations.} 

    Let $\mathcal{U} = \{U_1,\ldots,U_n\}$ be a collection of closed, convex sets.
    Let $\bigcup \mathcal{U}$ denote the union of the sets in $\mathcal{U}$, i.e.\ $\bigcup\mathcal{U} := \bigcup_{i=1}^n U_i$.
    We say that the set $\mathcal{U}$ is a \emph{cover} of the space $\bigcup \mathcal{U}$.
    The \emph{nerve} of $\mathcal{U}$, denoted $\nerve(\mathcal{U})$ is the abstract simplicial complex defined as
    \[
      \nerve(\mathcal{U}) : = \left\{I\subseteq [n] \mid \bigcap_{i\in I} U_i \neq \emptyset\right\}.
    \]
    This construction is illustrated in Fig.\ref{fig:nerve}.
    The Nerve Theorem~\cite[Cor. 4G.3]{hatcher01} implies that $\nerve(\mathcal{U})$ is homotopy equivalent to $\bigcup \mathcal{U}$.

    \begin{figure}
      \centering
      \appdx{\appendixMode}{
        \includegraphics[width=.40\textwidth]{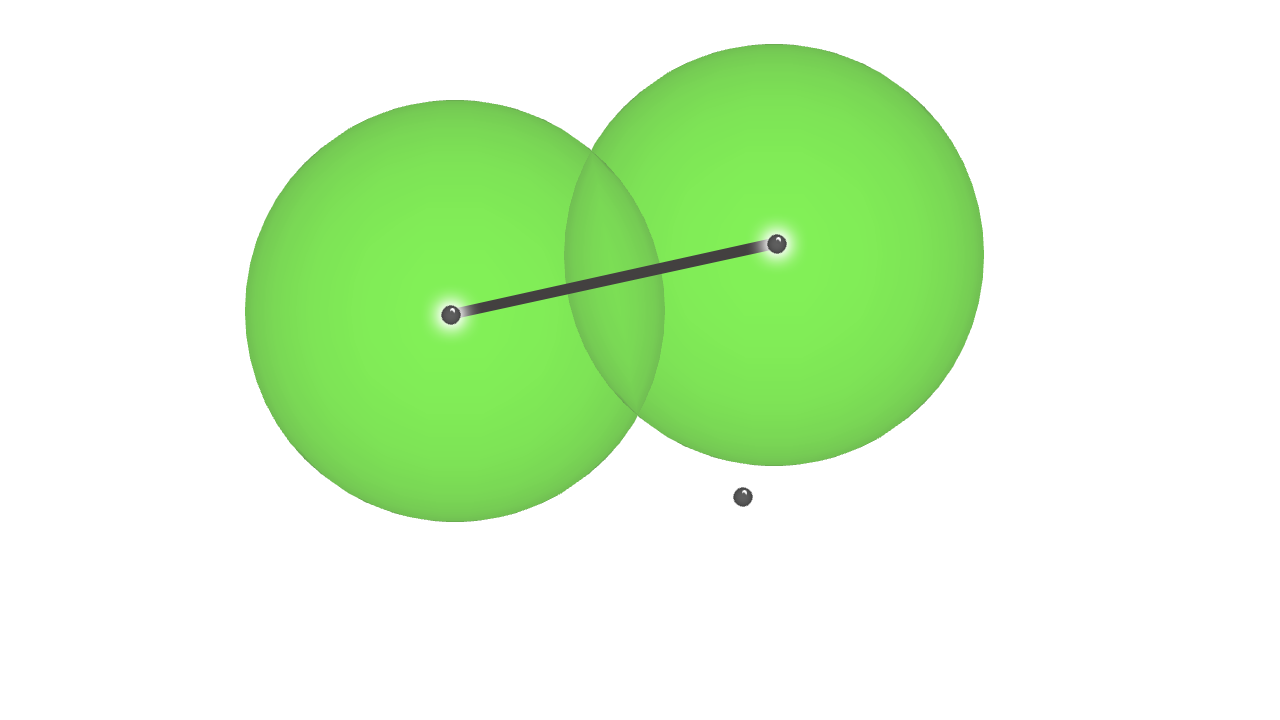}
        \includegraphics[width=.40\textwidth]{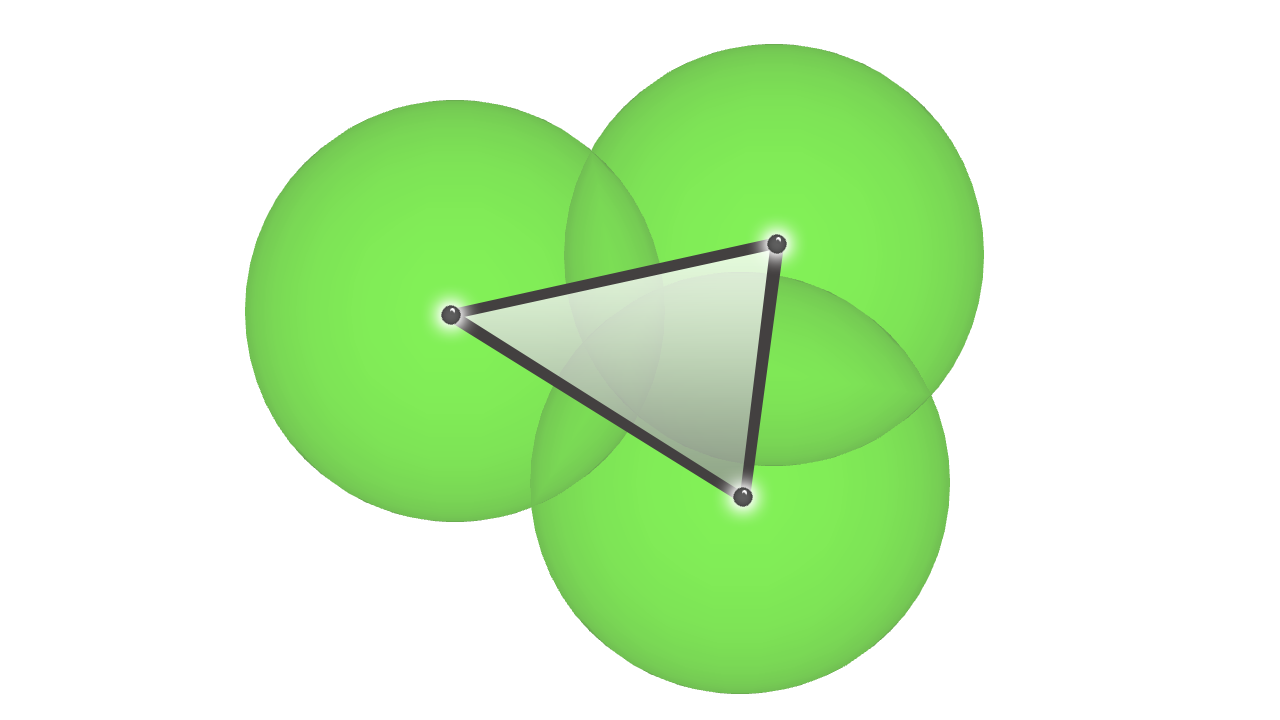}
      }{
        \includegraphics[width=.20\textwidth]{figures/nerve03}
        \includegraphics[width=.20\textwidth]{figures/nerve05}      
      }
      \caption{The nerve has an edge for each pairwise intersection, a triangle for each $3$-way intersection (right), etc.}
      \label{fig:nerve}
    \end{figure}

    Similarly, one can construct a nerve filtration from a cover of a filtration by filtrations.
    Specifically, let $\mathcal{U} = \{\{U_1^\alpha\},\ldots \{U_n^\alpha\}\}$ be a collection of filtrations parameterized by real numbers such that for each $i\in [n]$ and each $\alpha\ge 0$, the set $U_i^\alpha$ is closed and convex.
    As shorthand, we write $\mathcal{U}^\alpha$ to denote the set $\{U_1^\alpha,\ldots, U_n^\alpha\}$.
    As before, the Nerve Theorem implies that $\bigcup \mathcal{U}^\alpha$ is homotopy equivalent to $\nerve(\mathcal{U}^\alpha)$.
    The Persistent Nerve Lemma~\cite{chazal08towards} implies that the filtrations $\{\bigcup \mathcal{U}^\alpha\}_{\alpha\ge 0}$ and $\{\nerve(\mathcal{U}^\alpha)\}_{\alpha\ge 0}$ have identical persistent homology.


  \paragraph{\v Cech and Rips Filtrations.} 
  \label{par:_v_cech_and_rips_filtrations}
    A common filtered nerve is the \emph{\v Cech filtration}.
    It is defined as $\{\cech_\alpha(P)\}$, where
    \[
      \cech_\alpha(P) := \nerve\{\ball(p_i, \alpha) \mid i\in [n]\}.
    \]
    Notice that this is just the nerve of the cover of the $\alpha$-offsets by the $\alpha$-radius balls.
    Thus, the Persistent Nerve Lemma implies that $\{P^\alpha\}$ and $\{\cech_\alpha(P)\}$ have identical persistence barcodes.

    A similar filtration that is defined for any metric is called the \emph{(Vietoris-)Rips filtration} and is defined as $\{\rips_\alpha(P)\}$, where
    \[
      \rips_\alpha(P) := \{J\subseteq [n] \mid \max_{i,j\in J} \dist(p_i,p_j)\le 2\alpha\}.
    \]
    Note that if $\dist$ is the max-norm, $\ell_\infty$, then $\rips_\alpha(P) = \cech_\alpha(P)$.
    Moreover, because every finite metric can be isometrically embedded into $\ell_\infty$, every Rips filtration is isomorphic to a nerve filtration.


  \paragraph{Greedy Permutations.} 
    Let $P$ be a set of points in some metric space with distance $\dist$.
    A \emph{greedy permutation} of $P$ goes by many names, including landmark sets, farthest point sampling, and discrete center sets.
    We say that $P = \{p_1,\ldots,p_n\}$ is ordered according to a greedy permutation if each $p_i$ is the farthest point from the first $i-1$ points.
    We let $p_1$ be any point.
    Formally, let $P_i = \{p_1,\ldots,p_i\}$ be the $i$th \emph{prefix}.
    Then, the ordering is greedy if and only if for all $i\in\{2,\ldots,n\}$,
    \[
      \dist(p_i, P_{i-1}) = \max_{p\in P} \dist(p,P_{i-1}).
    \]
    For each point $p_i$, the value $\lambda_i := \dist(p_i, P_{i-1})$ is known as the \emph{insertion radius}.
    By convention, we set $\lambda_1 = \infty$.
    It is well-known (and easy to check) that $P_i$ is a $\lambda_i$-net in the sense that it satisfies the conditions:  for all distinct $p,q\in P_i$, $\dist(p,q) \ge \lambda_i$ (\textbf{packing}) and  $P\subseteq P_i^{\lambda_i}$ (\textbf{covering}).
                                                                                                                             

  \section{Perturbed Distances} 
\label{sec:perturbed_distances}

  A convenient first step in making a sparse version of the \v Cech filtration is to ``perturb'' the distance.
  Given a greedy permutation, we perturb the distance function so that as the radius increases, only a sparse subset of points continues to contribute to the offsets.
  This can most easily be viewed as changing the radius of the balls slightly so that some balls will be completely covered by their neighbors and thus will not contribute to the union.
  Fix a constant $\e<1$ that will control the sparsity.
  As we will show in Lemma~\ref{lem:covering}, at scale $\alpha$, there is an $\e\alpha$-net of $P$ whose perturbed offsets cover the perturbed offsets of $P$.
  Assuming the points $P=\{p_1,\ldots,p_n\}$ are ordered by a greedy permutation with insertion radii $\lambda_1,\ldots , \lambda_n$, we define the radius of $p_i$ at scale $\alpha$ as
  \[
    r_i(\alpha) :=
    \begin{cases}
      \alpha & \text{if } \alpha \le \lambda_i(1+\e)/\e \\
      \lambda_i(1+\e)/\e & \text{otherwise.}
    \end{cases}
  \]
  The \emph{perturbed $\alpha$-offsets} are defined as
  \[
    \pP^\alpha := \bigcup_{i\in [n]}\ball(p_i, r_i(\alpha)).
  \]
  To realize the sparsification as described, we want to remove balls associated with some of the points as the scale increases.
  This is realized by defining the $\alpha$-ball for a point $p_i\in P$ to be
  \[
    b_i(\alpha) :=
    \begin{cases}
      \ball(p_i,r_i(\alpha)) & \text{if } \alpha \le \lambda_i(1+\e)^2/\e \\
      \emptyset & \text{otherwise.}
    \end{cases}
  \]

  The usefulness of this perturbation is captured by the following covering lemma, which is depicted in the tops of the cones in Fig.~\ref{fig:cones_covered}.

  \begin{figure}
    \centering
    \appdx{\appendixMode}{
      \includegraphics[width=0.9\textwidth]{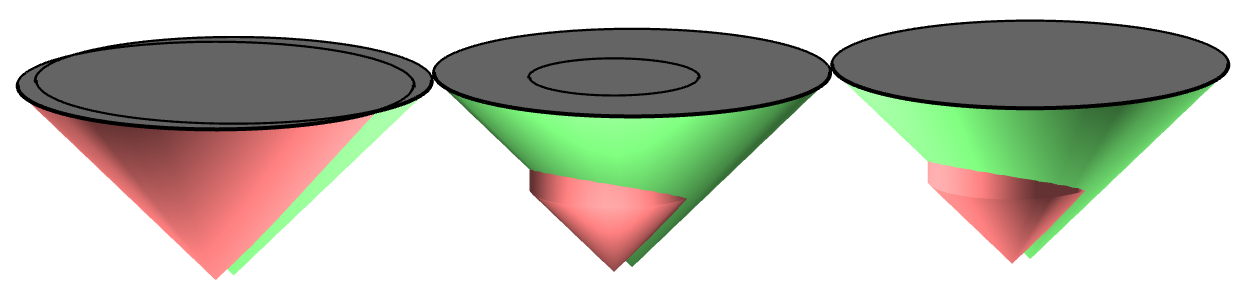}
    }{
      \includegraphics[width=0.45\textwidth]{figures/cones_covered}
    }
    \caption{Left: two growing balls trace out cones in one dimension higher.
    Center: One of the cones has a maximum radius.
    Right: Limiting the height of one cone guarantees that the top is covered.}
    \label{fig:cones_covered}
  \end{figure}

  \begin{lemma}[Covering Lemma]\label{lem:covering}
    Let $P = \{p_1,\ldots,p_n\}$ be a set of points ordered by a greedy permutation with insertion radii $\lambda_1,\ldots,\lambda_n$.
    For any $\alpha,\beta\ge 0$, and any $p_j\in P$, there exists a point $p_i\in P$ such that
    \begin{enumerate}
      \item if $\beta\ge \alpha$ then $b_j(\alpha)\subseteq b_i(\beta)$, and
      \item if $\beta\ge (1+\e)\alpha$, then $\ball(p_j,\alpha)\subseteq b_i(\beta)$.
    \end{enumerate}
  \end{lemma}
  \begin{proof}
    Fix any $p_j\in P$.
    We may assume that $\beta \ge \lambda_j(1+\e)^2/\e$, for otherwise, choosing $p_i = p_j$ suffices to satisfy both clauses, the first because $b_j(\alpha)\subseteq b_j(\beta)$ and the second because $\ball(p_j,\alpha) = b_j(\alpha) \subseteq b_j(\beta)$.
    This assumption is equivalent to the assumption that $b_j(\beta) = \emptyset$.

    By the covering property of the greedy permutation, there is a point $p_i\in P$ such that $\dist(p_i,p_j)\le \e\beta/(1+\e)$ and $\lambda_i\ge \e\beta/(1+\e)$.
    It follows that $r_i(\beta) = \beta$ and $b_i(\beta) = \ball(p_i, \beta)$.
    Recall that $\lambda_1 = \infty$ by convention, so $b_1(\beta)\neq \emptyset$, and for large values of $\beta$, choosing $p_i = p_1$ suffices.

    To prove the first clause, fix any point $x\in b_j(\alpha)$.
    By the triangle inequality,
    \begin{align*}
      \dist(x,p_i)
        & \le \dist(x, p_j) + \dist(p_i,p_j)
        \le r_j(\alpha) + \e\beta/(1+\e)\\
        & \le \lambda_j(1+\e)/\e + \e\beta/(1+\e)
        \le \beta
        = r_i(\beta).
    \end{align*}
    So, $x\in b_i(\beta)$ and thus, $b_j(\alpha)\subseteq b_i(\beta)$ as desired.

    To prove the second clause of the lemma, fix any $x\in \ball(p_j,\alpha)$.
    By the triangle inequality,
    \begin{align*}
      \dist(x, p_i)
        & \le \dist(x, p_j) + \dist(p_i, p_j) \le \alpha + \e\beta/(1+\e)\\
        & \le \beta/(1+\e) + \e\beta/(1+\e)= r_i(\beta).
    \end{align*}
    So, as before, $x\in b_i(\beta)$ and thus, $\ball(p_j,\alpha)\subseteq b_i(\beta)$ as desired.
  \end{proof}




  \begin{corollary}\label{cor:perturbed_covering}
    Let $P = \{p_1,\ldots,p_n\}$ be a set of points ordered by a greedy permutation with insertion radii $\lambda_1,\ldots,\lambda_n$.
    For all $\alpha\geq 0$, $\pP^\alpha= \bigcup_i b_i(\alpha)$ and $\pP^\alpha\subseteq P^\alpha\subseteq\pP^{(1+\e)\alpha}$.
  \end{corollary}

  \appdx{\appendixMode}{
  \begin{proof}
    We will first show that $\pP^\alpha= \bigcup_i b_i(\alpha)$.

    Fix any $\alpha\ge 0$.
    For all $j\in[n]$, $b_j(\alpha)\subseteq \ball(p_j, r_j(\alpha))$, so
    \begin{equation}\label{eq:Ub_subset_pP}
      \bigcup\limits_{j\in[n]}b_j(\alpha)\subseteq\bigcup \limits_{j\in[n]}\ball(p_j,r_j(\alpha))=\pP^\alpha.
    \end{equation}

    To show that $\pP=\ball(p_j,r_j(\alpha))\subseteq \bigcup \limits_{j\in[n]}b_j(\alpha)$, we have two cases. If $\alpha\leq\frac{\lambda_j(1+\e)^2}{\e}$, then $b_j(\alpha)=\ball(p_j,r_j(\alpha))$.
    Else $\alpha>\frac{\lambda_j(1+\e)^2}{\e}$, which implies that $r_j(\alpha)=\frac{\lambda_j(1+\e)}{\e}$.
    Let $\gamma=r_j(\alpha)$, which implies $r_j(\gamma)=\gamma$ and $\alpha>(1+\e)\gamma,$ so there exists $i$ such that $\ball(p_j,\gamma)\subseteq b_i(\alpha)$ and equivalently $\ball(p_j,r_j(\alpha))\subseteq b_i(\alpha)$.
    Thus,
    \begin{equation}\label{eq:pP_subset_Ub}
      \pP=\bigcup \limits_{j\in[n]}\ball(p_j,r_j(\alpha))\subseteq \bigcup \limits_{j\in[n]}b_j(\alpha).
    \end{equation}
    So \eqref{eq:Ub_subset_pP} and~\eqref{eq:pP_subset_Ub} imply that $\pP^\alpha= \bigcup_i b_i(\alpha)$.

    Now, we will prove that $\pP^\alpha\subseteq P^\alpha\subseteq\pP^{(1+\e)\alpha}$.
    \begin{equation}\label{eq:pP_subset_P}
      \pP=\bigcup \limits_{j\in[n]}\ball(p_j,r_j(\alpha))\subseteq \bigcup \limits_{j\in[n]}\ball(p_j,\alpha)=P^\alpha,
    \end{equation}
    because $r_j(\alpha)\leq \alpha$.
    Let $\beta=(1+\e)\alpha$, then for all $j\in[n]$ there exists $i$ such that $\ball(p_j,\alpha)\subseteq b_i(\beta)$ by statement 2 in Lemma ~\ref{lem:covering}, implying
    \begin{equation}\label{eq:Palpha_subset_pPbeta}
      P^\alpha=\bigcup \limits_{j\in[n]}\ball(p_j,\alpha)\subseteq \bigcup\limits_{j\in[n]}b_j(\beta)=\pP^\beta =\pP^{(1+\e)\alpha}.
    \end{equation}
    Thus \eqref{eq:pP_subset_P} and \eqref{eq:Palpha_subset_pPbeta} imply that $\pP^\alpha\subseteq P^\alpha\subseteq \pP^{(1+\e)\alpha}$
    \end{proof}
  }{A proof may be found in the full paper~\cite{cavanna15geometric_arxiv}.}
  Corollary \ref{cor:perturbed_covering} implies the following proposition using standard results on the stability of persistence barcodes ~\cite{chazal09proximity}.

  \begin{proposition}\label{prop:perturbed_offsets_are_close}
    The persistence barcode of the perturbed offsets $\{\pP^\alpha\}_{\alpha\ge 0}$ is a $(1+\e)$-approximation to the persistence barcode of the offsets $\{P^\alpha\}_{\alpha\ge 0}$.
  \end{proposition}


  \section{Sparse Filtrations} 
\label{sec:sparse_filtrations}

  The \emph{sparse \v Cech complex} is defined as $Q^\alpha := \nerve\{b_i(\alpha)  \mid i\in [n]\}$.
  Notice that because $b_i(\alpha) = \emptyset$ unless $\lambda_i$ is sufficiently large compared to $\alpha$, there are fewer vertices as the scale increases.
  This is the desired sparsification.
  Unfortunately, it means that the set of complexes $\{Q^\alpha\}$ is not a filtration, but this is easily remedied by the following definition.
  The \emph{sparse \v Cech filtration} is defined as $\{S^\alpha\}$, where
  \[
    S^\alpha := \bigcup_{\delta\le \alpha} Q^\delta = \bigcup_{\delta \le \alpha} \nerve\{b_i(\delta)  \mid i\in [n]\}.
  \]

  This definition makes it clear that the sparse complex is a union of nerves, but it not obvious that it has the same persistent homology as the filtration defined by the perturbed offsets $\pP^\alpha := \bigcup_{i}b_i(\alpha)$.
  For such a statement, it would be much more convenient if $\{S^\alpha\}$ was itself a nerve filtration rather than a union of nerves, in which case the Persistent Nerve Lemma could be applied directly.
  In fact, this can be done by adding an extra dimension corresponding to the filtration parameter extending the balls $b_i(\alpha)$ into the perturbed cone shapes
  \[
    U_i^\alpha:=\bigcup_{\delta \le \alpha} (b_i(\delta) \times \{\delta\}).
  \]
  These sets, depicted in Figs.~\ref{fig:cones_covered} and~\ref{fig:all_cones}, allow the following equivalent definition of the complexes in the sparse \v Cech filtration.
  \[
    S^\alpha := \nerve\left\{U_i^\alpha \mid i\in [n]\right\}.
  \]

  \appdx{\appendixMode}{
  \begin{proposition}\label{prop:convexity}
    If $\dist$ is a convex metric and $r_i$ is a concave function then $U_i^\alpha:=\bigcup_{\delta \le \alpha} (b_i(\delta) \times \{\delta\})$ is convex.
  \end{proposition}

  \begin{proof}
  Given two points $(a,\delta_a)$, $(b,\delta_b)\in U_i^\alpha$, $\dist(a,p_i)\leq r_i(\delta_a)$ and likewise $\dist(b,p_i)\leq r_i(\delta_b)$ by definition of $r_i$.
  Let $c=(1-t)a+tb$ and let $\delta_c=(1-t)\delta_a+t\delta_b$, for $t\in[0,1]$.
  Now we bound $\dist(c,p_i)$ as follows.
  \begin{align*}
  \dist(c,p_i)&\leq(1-t)\dist(a,p_i)+t\dist(b,p_i)\because{$\dist$ is convex}\\
    &\leq (1-t)r_i(\delta_a)+ tr_i(\delta_b) \\
    &\leq r_i(\delta_c) \because{$r_i$ is concave}
  \end{align*}
  Thus we can conclude that $(c,\delta_c)$, a convex combination of arbitrary $(a,\delta_a)$ and $(b,\delta_b)$, is in $U_i^\alpha$ and $U_i^\alpha$ is convex.
  \end{proof}
  }{}

  \begin{theorem}\label{thm:main_thm}
    The persistence barcode of the sparse nerve filtration $\{S^\alpha\}_{\alpha\ge 0}$ is a $(1+\e)$-approximation to the persistence barcode of the offsets $\{P^\alpha\}_{\alpha\ge 0}$.
  \end{theorem}
  \begin{proof}
    For all $i$, the set $U_i^\alpha$ is convex because $r_i$ is concave %
    \appdx{\appendixMode}{
    by Proposition~\ref{prop:convexity}}{
    (see the full paper~\cite{cavanna15geometric_arxiv} for a proof.)
    }
    It follows that the sets $U_i^\alpha$ satisfy the conditions of the Persistent Nerve Lemma.
    So, $\{S^\alpha\}$ has the same persistence barcode as the filtration $\{B^\alpha\}$, where $B^\alpha := \bigcup_{i}U_i^\alpha$.
    \begin{figure}[htbp]
      \centering
      \appdx{\appendixMode}{
        \includegraphics[width=0.9\textwidth]{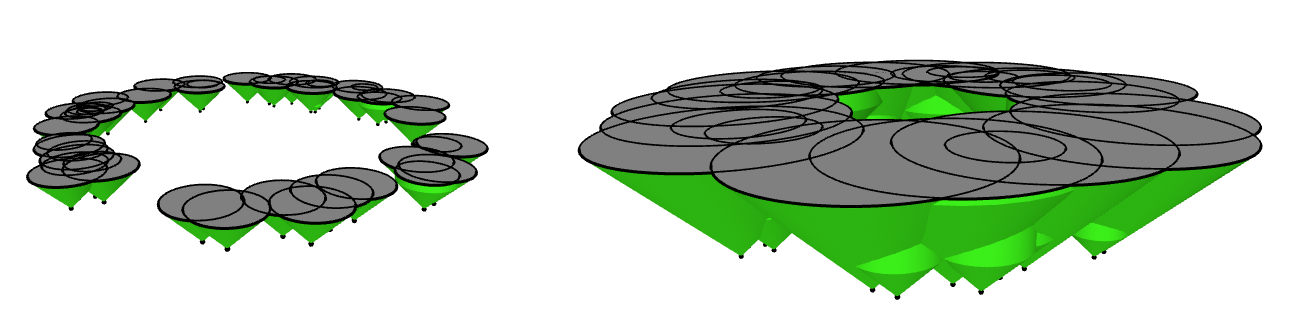}
      }{
        \includegraphics[width=0.45\textwidth]{figures/all_cones}
      }
      \caption{The collection of cones $B^\alpha$ at two different scales.
      The top of the cones is the union of (perturbed) balls.}
      \label{fig:all_cones}
    \end{figure}

    The Covering Lemma implies that the linear projection of $B^\alpha$ to $\pP^\alpha$ that maps $(x,\delta)$ to $x$ is a homotopy equivalence as each fiber is simply connected.
    Moreover, the projection clearly commutes with the inclusions $B^\alpha\hookrightarrow B^\beta$ and $\pP^\alpha\hookrightarrow \pP^\beta$, from which, it follows that $\pers\{\pP^\alpha\} = \pers\{B^\alpha\} = \pers\{S^\alpha\}$.
    So, the claim now follows from Proposition~\ref{prop:perturbed_offsets_are_close}.
  \end{proof}

  \section{Algorithms} 
\label{sec:algorithms}  
  In previous work, it was shown how to use metric data structures~\cite{har-peled06fast} to compute the sparse Rips filtration in $O(n\log n)$ time~\cite{sheehy13linear} when the doubling dimension is constant.
  The same approach also works for the sparse nerve filtrations described here.
  However, it depends on the construction of a net-tree~\cite{har-peled06fast}, which is an intricate data structure.  
  
  In this section, we present a simpler technique to construct a sparse nerve filtration from a greedy permutation of a finite metric $(P,\dist)$.
  Throughout, we assume that the doubling dimension of $(P,\dist)$ is constant.
  We show how to construct a sparse nerve filtration in linear time from the greedy permutation. 
  Our approach starts with finding all edges and their birth times.

  Let $G$ be a directed graph whose vertices are the points of $P$ and whose edges are the edges of the sparse nerve filtration of $P$ directed from smaller to larger insertion radius.
  In Section~\ref{sub:finding_neighborhoods_suffices}, it is shown that for each directed edge $(p_i,p_j)$ in $G$, $\dist(p_i,p_j)\le \kappa \lambda_i$, for a constant $\kappa$.
  This reduces the problem of finding the edges of the filtration to the problem of finding points in a given neighborhood.
  Moreover, we show that the out-degree of a vertex in $G$ is constant.
  Then, in Section~\ref{sub:how_to_find_neighborhoods_using_a_greedy_permutation}, we present an algorithm to construct $G$ from the greedy permutation and show that it runs in linear time.
  Finally, in Section~\ref{sub:higher_dimensional_simplices}, we give an algorithm for building higher dimensional simplices using the directed graph and bound its running time.
 
\subsection{Finding Neighborhoods Suffices} 
\label{sub:finding_neighborhoods_suffices}
  The vertices adjacent to $p_i$ in the directed graph $G$ are the points $p_j$ with insertion radius at least that of $p_i$ such that their corresponding balls intersect at some scale $\alpha$.
  The following lemma shows that these points have distance at most a constant times $\lambda_i$ to $p_i$.
  Then, Lemma~\ref{lem:constant_size_neighborhoods} will use this fact to show that the number of adjacent vertices is at most a constant.
  
  \begin{lemma}\label{lem:edges}
    For a given point $p_i$ with insertion radius $\lambda_i$ in the directed graph $G$, all adjacent points to $p_i$ are located in a $\ball(p_i,\kappa \lambda_i)$, where $\kappa=\frac{\e^2+3\e+2}{\e}$ and $\e>0$.
  \end{lemma}
  \begin{proof}
    In the directed graph $G$, a vertex $p_j$ is adjacent to vertex $p_i$ if $\lambda_i\le \lambda_j$ and for some scale $\alpha$, $b_i(\alpha)\cap b_j(\alpha)\ne \emptyset$.
    These balls intersect before $p_i$ disappears, so
    \[
      b_i(\lambda_i(1+\e)^2/\e)\cap b_j(\lambda_i(1+\e)^2/\e)\ne \emptyset.
    \]
    The distance between $p_i$ and $p_j$ is bounded as follows.
    \begin{align*}
      \dist(p_i,p_j)
      & \le r_i(\lambda_i(1+\e)^2/\e)+r_j(\lambda_i(1+\e)^2/\e)\\
      & \le \lambda_i(1+\e)/\e+\lambda_i(1+\e)^2/\e\\
      & \le \frac{\e^2+3\e+2}{\e} \lambda_i.
    \end{align*}
    Thus, all adjacent vertices to $p_i$ lie in a ball with center $p_i$ and radius $\kappa\lambda_i$.
  \end{proof}
  
  \begin{lemma}\label{lem:constant_size_neighborhoods}
    For a point set $P$ ordered by a greedy permutation and with doubling dimension $\rho$, each $p_i\in P$ has $\kappa^{O(\rho)}$ neighbors in the directed graph $G$, where $\kappa=\frac{\e^2+3\e+2}{\e}$ and $\e>0$.
  \end{lemma}
  \begin{proof}
    The proof uses a standard packing argument.
    By the definition of the doubling dimension, $\ball(p_i, \kappa\lambda_i)$ can be covered by $2^{\rho (\lceil\lg\kappa\rceil + 1)}$ balls of radius less than $\lambda_i$.
    Since the neighbors are pairwise $\lambda_i$-separated, there can be at most one point in each such ball.
    Therefore, the number of balls is $2^{O(\rho \lg \kappa)}=\kappa^{O(\rho)}$.
  \end{proof}

  
\subsection{How to find neighborhoods using a greedy permutation} 
  \label{sub:how_to_find_neighborhoods_using_a_greedy_permutation}
  In this section, we construct the directed graph $G$ as described from a given greedy permutation.
  In Section~\ref{sub:finding_neighborhoods_suffices}, it was shown that to construct $G$ it suffices to find points within a metric ball around each point.
  We build an efficient data structure to maintain these points.
  
  Let $P= \{p_1,\ldots, p_n\}$ be the points in $(P,\dist)$ ordered according to a greedy permutation.
  For each $p_i\in P$, let $\pred(p_i)\in P_{i-1}$ denote the nearest point to $p_i$ among the first $i-1$ points in the ordering.
  So, the insertion radius of $p_i$ is $\lambda_i = \dist(p_i,\pred(p_i))$.
  The \emph{level} of $p_i$ is defined as $\ell_i := \lceil \lg \lambda_i \rceil$.

  The goal is to process the points one at a time in the greedy ordering, and for each $p_i$, to find all preceding points within distance $\kappa\lambda_i$, where $\kappa=(\e^2+3\e+2)/\e$ and $\e>0$ is a fixed constant chosen by the user.
  Because all neighbors of $p_i$ in a sparse nerve filtration have this property by Lemma~\ref{lem:edges}, we can use this list to find all the neighbors.

  We will define a data structure $\D$ used to extract neighborhood information in the directed graph $G$.
  For each point $p_i$ in $P$, $\D$ stores $\pred(p_i)$, $\ell_i$, and three other pieces of information:
  \begin{enumerate}
    \item a point $\parent(p_i)$ called the \emph{parent},
    \item a list of points $\nbr(p_i)$ called the \emph{neighbors}, and
    \item a list of points $\ch(p_i)$ called the \emph{children} of $p_i$.
  \end{enumerate}
  These three objects change over the course of the algorithm.
  We only require that for all $i\in [n]$ and all $p_j\in P_i$, they satisfy the following invariants after $i$ points have been processed.
  \begin{enumerate}
    \item\textbf{Parent Invariant:} $\parent(p_j) = p_j$ if $\ell_j>\ell_i$. 
    Otherwise, $\parent(p_j)$ is a point $p_k$ such that $\ell_k>\ell_i$ and $\dist(p_j, p_k)\le 2^{\ell_i}$.
    \item\textbf{Child Invariant:} $\ch(p_j) \supseteq \{p_j\} \cup \{p_k\in P_i \mid \parent(p_k) = p_j \text{ and } \ell_k = \ell_i\}$.        
    \item\textbf{Neighbor Invariant:} $\nbr(p_j) \supseteq \{p_k\in P_i \mid \dist(p_j, p_k)\le \kappa2^{\min\{\ell_j,\ell_k, \ell_i+1\}}\}$.
  \end{enumerate}
  The second invariant states that the children list of $p_j$ contains all points at the same level as $p_i$ that have $p_j$ as a parent.
  The third invariant says that the neighbor lists contain all nearby points where ``nearby'' is related to the insertion radius of $p_i$.
  This last invariant implies the correctness of the algorithm, because for $j=i$, it says the neighbor list contains the set we are interested in.
  We maintain the lists for the other points to help us do updates at each step.

  Furthermore, we assume that $\D$ provides constant-time access to the list of points in a specific level.

  Algorithm~\ref{alg:ds_for_edges} shows how a new point $p_i$ can be inserted into the data structure $\D$. 
  In fact, we process points of a greedy permutation one by one and after inserting a new point in $\D$, we update the directed graph $G$, which is used to extract higher dimensional simplices.
  \begin{algorithm}
    \caption{Inserting a new point into the data structure $\D$}
    \label{alg:ds_for_edges}
    \begin{algorithmic}[1]
      \Procedure{Insert}{$\D,p_i$}
        \If{$\ell_i< \ell_{i-1}$}
          \ForAll{$p_k$ such that $\ell_k= \ell_{i-1}$}
            \State $\parent(p_k) \gets p_k$
          \EndFor
        \EndIf  
        \State $p_j\gets \pred(p_i)$
        \State $\parent(p_i) \gets \parent(p_j)$
        \ForAll{$p_k\in \nbr(\parent(p_j))$}
          \If{$\dist(p_i,p_k)\le \dist(p_i,\parent(p_i))$ and $\ell_k> \ell_i$}
            \State $\parent(p_i)\gets p_k$
          \EndIf
        \EndFor
        \State add $p_i$ to $\ch(p_i)$
        \State add $p_i$ to $\ch(\parent(p_i))$
        \State add $p_i$ to $\nbr(p_i)$
        \ForAll{$p_k\in \ch(\nbr(\parent(p_i)))$}
          \If{$\dist(p_i,p_k)\le\kappa 2^{\ell_i}$}
            \State add $p_k$ to $\nbr(p_i)$
            \State add $p_i$ to $\nbr(p_k)$
          \EndIf
        \EndFor
      \EndProcedure
    \end{algorithmic}
  \end{algorithm}
  
  \begin{lemma}
    Let $P = (p_1,\ldots, p_n)$ be a greedy permutation.
    For all $i\in \{2,\ldots,n\}$, if $\D$ is a data structure on $P_{i-1}$ satisfying the three invariants, then it also satisfies the invariants after calling \textsc{Insert}$(\D,p_i)$.
  \end{lemma}
  \begin{proof}
    We consider the invariants one at a time.
    
    First, if $\ell_i< \ell_{i-1}$, the algorithm updates the parents of all nodes in level $\ell_{i-1}$.
    Note that these are the only points required to be updated to satisfy the Parent Invariant for all points in $P_{i-1}$.

    Next, we check that there exists a point $p_k$ such that setting $\parent(p_i)$ to $p_k$ satisfies the Parent Invariant.
    The algorithm iterates over $\nbr(\parent(p_j))$ to find the closest point with a level higher than $\ell_i$.
    We first show there exists a point in a higher level that satisfies the Parent Invariant and then show that any such point is in $\nbr(\parent(p_j))$.
    Let $z=\argmax_{z<i}\{\ell_z\mid \ell_z>\ell_i\}$.
    Let $p_k$ be the closest point in $P_z$ to $p_i$.
    So, 
    \begin{align*}
      \dist(p_i,p_k)=d(p_i,P_z)\le \max_{p\in P}\dist(p,P_z)=\lambda_{z+1}\le 2^{\ell_{z+1}}\le 2^{\ell_i}.
    \end{align*}
    Thus, some point $p_k$ could satisfy the Parent Invariant. 
    Any such point $p_k$ satisfies   
    \begin{align*}
      \dist(p_k, \parent(p_j))
      &\le \dist(p_k, p_i) + \dist(p_i, p_j) + \dist(p_j, \parent(p_j))\\
      &\le 2^{\ell_i} + \lambda_i + 2^{\ell_i}\\
      &< 2^{\ell_i} + 2^{\ell_i} + 2^{\ell_i}\\
      &=\frac{3}{2}\cdot2^{\ell_{i}+1}\\
      &<\kappa2^{\ell_i+1}.
    \end{align*}
    Therefore, $p_k\in \nbr(\parent(p_j))$ by the Neighbor Invariant.

    For the Child Invariant, $p_i$ needs to be inserted into $\ch(\parent(p_i))$.
    No other children lists need to change to satisfy the invariant.
    
    Next, to satisfy the Neighbor Invariant, neighbor lists should be updated.
    This only involves finding the neighbor list of $p_i$ and also adding $p_i$ to the neighbor lists of its neighbors.
    For this step, it suffices to check that if $p_k$ must be added to $\nbr(p_i)$, i.e.\ if $\dist(p_i,p_k)\le \kappa2^{\ell_i}$, then $p_k\in \ch(\nbr(\parent(p_i)))$.
    That is, the neighbors of $p_i$ are all children of neighbors of the parent of $p_i$.
    This follows from the triangle inequality and the invariants for $i-1$ as follows.
    \begin{align*}
      \dist(\parent(p_k), \parent(p_i))
      &\le \dist(\parent(p_k), p_k) + \dist(p_k, p_i) + \dist(p_i, \parent(p_i))\\
      &\le 2^{\ell_i} + \kappa2^{\ell_i} + 2^{\ell_i}\\
      &= (1+\kappa/2)2^{\ell_i+1}\\
      &< \kappa2^{\ell_i+1}.
    \end{align*}
    So, it follows that $\parent(p_k)\in \nbr(\parent(p_i))$, and so $p_k\in \ch(\parent(p_k))\subseteq \ch(\nbr(\parent(p_i)))$.
    If $p_k$ is added to $\nbr(p_i)$, then it is required to add $p_i$ to $\nbr(p_k)$ and the algorithm does this.
  \end{proof}
  
  Algorithm~\ref{alg:directed_graph} constructs all edges that appear in a sparse filtration. 
  It receives a set of points $P$, which is ordered by a greedy permutation, as input and returns a directed graph $G$.
  As we mentioned earlier, we will use the directed graph $G$ to find higher dimensional simplices.
  For each point $p_i$, the algorithm invokes the \textsc{Insert} procedure to find its neighbors.
  Then, to build sparse edges between $p_i$ and its neighbors, Algorithm~\ref{alg:checkedge} is called.
  If an edge appears in the sparse filtration, \textsc{EdgeBirthTime} method returns the birth time of the edge and $\infty$ otherwise.
  Finally, for an edge in the sparse filtration, a directed edge from $p_i$ to $p_j$ will be inserted into $G$.

  \begin{algorithm}
    \caption{Constructing edges of a sparse filtration}
    \label{alg:directed_graph}
    \begin{algorithmic}[1]
      \Procedure{ConstructEdges}{$P=\{p_1,\ldots,p_n\}$}
        \State initialize $\D$ with $p_1$ \Comment{adds $p_1$ to $\ch(p_1)$ and sets $\parent(p_1)=p_1$.}
        \State initialize a directed graph $G$ on $P$
        \For{$i=2$ to $n$}
          \State \textsc{Insert($\D,p_i$)}
          \ForAll{$p_j\in \nbr(p_i)$}
            \State $\alpha\gets$\textsc{EdgeBirthTime($p_i,p_j$)}
            \If{$\alpha<\infty$}
              \State add a directed edge from $p_i$ to $p_j$ with birth time $\alpha$ to $G$
            \EndIf
          \EndFor
        \EndFor
        \Return $G$
      \EndProcedure
    \end{algorithmic}
  \end{algorithm}
  
  \begin{algorithm}
    \caption{Compute the birth time of an edge}
    \label{alg:checkedge}
    \begin{algorithmic}[1]
      \Procedure{EdgeBirthTime}{$p_i,p_j$}
        \If{$\lambda_i>\lambda_j$}
          \State swap $p_i$ and $p_j$
        \EndIf
        \If{$\dist(p_i,p_j)\le\frac{2\lambda_i(1+\e)}{\e}$}
          \State \textbf{return} $\frac{\dist(p_i,p_j)}{2}$
        \EndIf
        \If{$\dist(p_i,p_j)\le\frac{(\lambda_i+\lambda_j)(1+\e)}{\e}$}
          \State \textbf{return} $\dist(p_i,p_j)-\frac{\lambda_i(1+\e)}{\e}$
        \EndIf
        \State \textbf{return} $\infty$
      \EndProcedure
    \end{algorithmic}
  \end{algorithm}
  
  \begin{theorem}
    Given a greedy permutation of a finite metric $(P,\dist)$ of constant doubling dimensions and the nearest predecessors $\pred(p)$ for each $p\in P$, one can compute the edges of the sparse nerve filtration of $(P,\dist)$ in $O(n)$ time.
  \end{theorem}
    \begin{proof}
    Algorithm~\ref{alg:directed_graph} finds all edges in a sparse filtration.
    The running time of this algorithm mainly depends on the running time of \textsc{Insert} procedure and the size of neighbor list for each point.
    
    In Algorithm~\ref{alg:ds_for_edges}, the most common operation for the lists $\nbr(p_j)$ and $\ch(p_j)$ is to enumerate their elements.
    Any time a list is enumerated, we can check each point in constant time to see if it is still required to satisfy the invariant and remove it otherwise.
    Note that although the invariants only specify a subset that must appear, it is easy to check that enumerating these lists can be done in amortized constant time.
    This follows from two facts.
    First, the required subsets have constant size (by standard packing arguments).
    Second, the number of removals is at most the number of insertions, so we charge the cost of visiting such a point in the enumeration to the cost of its insertion. 
    
    In addition, when inserting $p_i$, if $\ell_i< \ell_{i-1}$, then $parent(p_k)$ is updated for all $p_k$ such that $\ell_k = \ell_{i-1}$.
    The total cost of such operations is $O(n)$ as no parent is updated twice.
    
    After insertion of a point $p_i$ into $\D$, Algorithm~\ref{alg:checkedge} is called for all points in $\nbr(p_i)$ to check whether an edge belongs to the sparse filtration.
    This algorithm has a constant running time.
    In addition, by Lemma~\ref{lem:constant_size_neighborhoods}, the size of a neighbor list for each point is constant.
    Therefore, for each point, the cost of finding these edges in the sparse filtration in $O(1)$.
  \end{proof}


\subsection{Higher Dimensional Simplices} 
\label{sub:higher_dimensional_simplices}
  In the previous section, it is shown that from a greedy permutation, the edges of a sparse nerve filtration can be constructed in linear time.
  Now, we present an algorithm to find $k$-simplices in the sparse filtration for $k>1$.
  As mentioned earlier, the directed graph $G$ built from the edges of the sparse nerve filtration will be used to construct higher dimensional simplices.
  
  Let $E(v)$ be the vertices adjacent to a vertex $v$ in $G$ (for each $u\in E(v)$, there is a directed edge from $v$ to $u$).
  To find a $k$-simplex for $k>1$ containing a vertex $v$, we consider all subsets $\{u_1,\ldots,u_k\}$ of $k$ vertices in $E(v)$.
  If $\{v,u_1,\ldots,u_k\}$ forms a $(k+1)$-clique, we check the clique to see whether it creates a $k$-simplex and compute its birth time.
  The birth time of a $k$-simplex $\sigma$ in a nerve filtration is defined as follows.
  \begin{align*}
    \textsc{SimplexBirthTime}(\sigma) &:= \min\left\{\alpha:\bigcap_{j\in\sigma}U_j^\alpha \ne \emptyset\right\}
    &= \min\left\{\alpha:\bigcap_{j\in\sigma}b_j(\alpha)\ne \emptyset\right\}.
  \end{align*}
  If no such $\alpha$ exists, then we define the birth time to be $\infty$.
  We assume the user provides a method, \textsc{SimplexBirthTime}, to compute birth times for their metric that runs in time polynomial in $k$.
  This function takes a $(k+1)$-clique as input.
  If at some scale $\alpha$, the corresponding balls have a common intersection, it returns the minimum such $\alpha$, otherwise, it returns $\infty$ indicating the $(k+1)$-clique is not a $k$-simplex in the sparse filtration.

  For the case of Rips filtrations (i.e.\ $\ell_\infty$), \textsc{SimplexBirthTime}$(\sigma)$ just needs to compute the maximum birth time of the edges and compare it to $\min_{p_i\in \sigma}\lambda_i(1+\e)^2/\e$ (the first time $t$ after which some $p_i\in\sigma$ has $b_i(t) = \emptyset$).
  For $\ell_2$, the corresponding computation is a variation of the minimum enclosing ball problem.

  Algorithm~\ref{alg:rips} finds the $k$-simplices and birth times in a sparse filtration.
  In this algorithm, $G$ is the given directed graph and the output $S$ is the set of pairs $(\sigma, t)$, where $\sigma$ is a $k$-simplex and $t$ is its birth time.
  \begin{algorithm}
    \caption{Find all $k$-simplices and birth times}
    \label{alg:rips}
    \begin{algorithmic}[1]
    \Procedure{FindSimplices}{$G,k$}
    \State $S\gets \emptyset$
    \ForAll{vertex $v$ in $G$}
    \ForAll{$\{u_1,\ldots,u_k\}\subseteq E(v)$}
    \If{$\{v,u_1,\ldots,u_k\}$ is a $(k+1)$-clique}
    \State $\sigma \gets \{v,u_1,\ldots,u_k\}$
    \State $t\gets $\textsc{SimplexBirthTime}$(\sigma)$
    \If{$t< \infty$}
    \State $S\gets S\cup (\sigma, t)$
    \EndIf
    \EndIf
    \EndFor
    \EndFor
    \State \textbf{return} $S$
    \EndProcedure
    \end{algorithmic}
  \end{algorithm}
  
  \begin{theorem}
    Given the edges of a sparse nerve filtration, Algorithm~\ref{alg:rips} finds the $k$-simplices of $\{S^\alpha\}$ in $\kappa^{O(k\rho)}n$ time, where $\rho$ is the doubling dimension of the input metric, $\kappa=(\e^2+3\e+2)/\e$, and $\e>0$.
  \end{theorem}
  \begin{proof}
    In Algorithm~\ref{alg:rips}, for every vertex $v$ in the directed graph $G$, there are ${|E(v)| \choose k}$ subsets with size $k$.
    In addition, by Lemma~\ref{lem:constant_size_neighborhoods}, $|E(p_i)| = \kappa^{O(\rho)}$.
    Therefore, the total running time of this algorithm will be $\kappa^{O(k\rho)}n$.
  \end{proof}
  

  \section{Removing Vertices} 
\label{sec:removing_vertices}

  Because the sparse filtration is a true filtration, no vertices are removed.
  When the cone is truncated, no new simplices will be added using that vertex, but it is still technically part of the filtration.
  The linear-size guarantee is a bound on the total number of simplices in the complex.
  Thus, by using methods such as zig-zag persistence or simplicial map persistence to fully remove these vertices when they are no longer needed cannot improve the asymptotic performance.
  Still, it may be practical to remove them (see~\cite{botnan15approximating}).
  A full theoretical or experimental analysis of the cost tradeoff of using a heavier algorithm to do vertex removal is beyond the scope of this paper.

  In this section, we show that the geometric construction leads to a natural choice of elementary simplicial maps (edge collapses) which all satisfy the so-called link condition.
  In the persistence by simplicial maps work of Dey et al.~\cite{dey14computing} and Boissonat et al.~\cite{boissonat13compressed}, a key step in updating the data structures to contract an edge is to first add simplices so that the so-called Link Condition is satisfied.
  The \emph{link} of a simplex $\sigma$ in a complex $K$ is defined as \appdx{\appendixMode}
  {\[                                                                                       
      \link\sigma = \{\tau \setminus \sigma \mid \tau\in K \text{ and } \sigma\subseteq\tau\}.
  \]}
  {$\link\sigma = \{\tau \setminus \sigma \mid \tau\in K \text{ and } \sigma\subseteq\tau\}$.}                                                                                          
  That is, the link $\sigma$ is formed by removing the vertices of $\sigma$ from each of its cofaces.
  An edge $\{u,v\}\in K$ satisfies the \emph{Link Condition} if and only if
  \[
    \link\{u,v\} = \link\{u\} \cap \link\{v\}.
  \]
  Dey et al.~\cite{dey99topology} proved that edge contractions induce homotopy equivalences when the link condition is satisfied.
  Thus, it gives a minimal local condition to guarantee that the contraction preserves the topology.
  More recently, it was shown that such a contraction does not change the persistent homology~\cite{dey14computing}.

  \begin{proposition}
    If $(P,\dist)$ is a finite subset of a convex metric space and $\{S^\alpha\}$ is its corresponding sparse filtration, then the last vertex $p_n$ has a neighbor $p_i$ such that the edge $\{p_n,p_i\}\in S^{\alpha}$ satisfies the link condition, where $\alpha = \lambda_n(1+\e)^2/\e$ and $\lambda_n$ is the insertion radius of $p_n$.
  \end{proposition}
  \begin{proof}
    It follows directly from the definition of a link that $\link\{u,v\}\subseteq \link\{u\}\cap\link\{v\}$ for all edges $\{u,v\}$.
    By the Covering Lemma (Lemma~\ref{lem:covering}), we know that there exists a $p_i\in P$ such that $b_n(\alpha)\subseteq b_i(\alpha)$.
    Thus, it suffices to check that $\link\{i\}\cap\link\{n\} \subseteq \link\{i,n\}$.
    Because the vertices are ordered according to a greedy permutation, $\lambda_n \ge \lambda_j$ for all $p_j\in P$.
    It follows that a simplex $J\in S^\alpha$ if and only if $\bigcap_{i\in J} b_j(\alpha)\neq \emptyset$.

    Let $J$ be any simplex in $\link\{i\}\cap\link\{n\}$.
    So, $i,n\notin J$ and $\bigcap_{j\in J \cup \{n\}}b_j(\alpha)\neq \emptyset$.
    Because $b_n(\alpha)\cap b_i(\alpha) = b_n(\alpha)$, it follows that $\bigcap_{j\in J \cup \{i,n\}}b_j(\alpha)\neq \emptyset$.
    Thus, we have $J\in \link\{i,n\}$ as desired.
  \end{proof}

  \section{Conclusion} 
\label{sec:conclusion}

In this paper, we gave a new geometric perspective on sparse filtrations for topological data analysis that leads to a simple proof of correctness for all convex metrics.
By considering a nerve construction one dimension higher, the proofs are primarily geometric and do not require explicit construction of simplicial maps.
This geometric view clarifies the non-zig-zag construction, while also showing that removing vertices can be accomplished with simple edge contractions.

 {
 \small
  \bibliographystyle{abbrv}
  \bibliography{bibliography}
 }
  \ifthenelse{\equal{\appendixMode}{1}}{\clearpage\appendix\input{appendix}}{}
  
\end{document}